\pdfoutput=1
\documentclass[12pt, draftclsnofoot, onecolumn]{IEEEtran}
\usepackage{graphicx}
\usepackage{algorithm}
\usepackage{algorithmic}

\usepackage{booktabs}
\usepackage[colorlinks,
linkcolor=black,
anchorcolor=black,
citecolor=black,
urlcolor = blue
]{hyperref}
\usepackage{amssymb}
\usepackage{subfigure}
\usepackage{multicol}
\usepackage{multirow}
\usepackage{amsmath}
\usepackage{bm}
\usepackage{cite}
\usepackage{gensymb}
\usepackage{amsfonts}
\usepackage{mathrsfs}
\usepackage{amsmath}
\usepackage{color}
\usepackage{balance}
\usepackage{bm}
\usepackage{setspace}
\usepackage{stfloats}
\usepackage{float}
\usepackage{epstopdf}
\usepackage{array}
\newcommand{\PreserveBackslash}[1]{\let\temp=\\#1\let\\=\temp}
\newcolumntype{C}[1]{>{\PreserveBackslash\centering}p{#1}}
\usepackage[flushleft]{threeparttable}
\usepackage{stfloats}
\usepackage[numbers,sort&compress]{natbib}

\newtheorem{lemma}{Lemma}

\newtheorem{proposition}{\it Proposition}
\newenvironment{proof}{\paragraph{Proof:}}{\hfill$\square$}

\begin{document}

\bibliographystyle{IEEEtran}

\title{Optimization on Multi-User Physical Layer Security of Intelligent Reflecting Surface-Aided VLC}

\author{Shiyuan Sun, Fang Yang, \emph{Senior Member, IEEE},	Jian Song, \emph{Fellow, IEEE}, and Zhu Han, \emph{Fellow, IEEE}	
\vspace{-0.3cm}

\thanks{
	This work was supported by the National Natural Science Foundation of China (61871255) and the Fok Ying Tung Education Foundation.
	\textit{(Corresponding author: Fang Yang)}
	
	Shiyuan Sun, Fang Yang, and Jian Song are with the Department of Electronic Engineering, Beijing National Research Center for Information Science and Technology, Tsinghua University, Beijing 100084, China, and also with the Key Laboratory of Digital TV System of Guangdong Province and Shenzhen City, Research Institute of Tsinghua University in Shenzhen, Shenzhen 518057, China 
	(e-mail: sunsy20@mails.tsinghua.edu.cn; fangyang@tsinghua.edu.cn; jsong@tsinghua.edu.cn).
	
	Zhu Han is with the Department of Electrical and Computer Engineering in the University of Houston, Houston, TX 77004 USA, and also with the Department of Computer Science and Engineering, Kyung Hee University, Seoul, South Korea, 446-701 
	(e-mail: zhan2@uh.edu).
}
}

\maketitle
\begin{abstract}
This letter investigates physical layer security in intelligent reflecting surface (IRS)-aided visible light communication (VLC).
Under the point source assumption, we first elaborate the system model in the scenario with multiple legitimate users and one eavesdropper, where the secrecy rate maximization problem is transformed into an assignment problem by objective function approximation.
Then, an iterative Kuhn-Munkres algorithm is proposed to optimize the transformed problem, and its computational complexity is in the second-order form of the numbers of IRS units and transmitters.
Moreover, numerical simulations are carried out to verify the approximation performance and the VLC secrecy rate improvement by IRS.
\end{abstract}


\IEEEpeerreviewmaketitle

\vspace{-0.3cm}
\section{Introduction}
\vspace{-0.1cm}
\label{Sec:Intro}
As an indispensable component of future wireless communication technologies, visible light communication (VLC) has long been concerned by authoritative institutions such as the VLC Consortium (VLCC) and hOME Gigabit Access project (OMEGA)~\cite{pathak2015visible}, and its scientific research and industrialization progress are deepening.
Generally, VLC shows outstanding advantages such as abundant frequency bandwidth, license-free merit, and low equipment cost~\cite{pathak2015visible}.
Nevertheless, the physical layer security of VLC systems is of paramount importance to be investigated.
Traditional VLC physical security is typically guaranteed by effective beamforming schemes and/or jamming techniques at transmitters~\cite{9317699}, and the upper and lower bounds of secrecy capacity have been derived in~\cite{mostafa2015physical}.

Fortunately, the emerging intelligent reflecting surface (IRS) technology provides a brand new perspective that the system security can be enhanced by actively re-directing the reflected signals.
The physical basis of IRS lies in the manipulation of electromagnetic waves, and mainstream hardware architectures and channel models in the visible light range have been investigated in~\cite{shirmanesh2020electro,najafi2019intelligent,2021intelligent,abdelhady2020visible}.
As a sequel, there is a growing body of literature that studies the performance of IRS-aided VLC systems, including the outage probability reduction in mobile free-space optical communication~\cite{wang2020performance}, the blockage problem mitigation in single-transmitter single-user~\cite{9543660} and multi-transmitter multi-user scenarios~\cite{sun_CL}, etc.
In the area of VLC physical layer security, a secure IRS-aided VLC model is established in~\cite{qian2021secure}, endeavoring to maximize the secrecy rate with the facilitation of IRS. 
A particle swarm optimization (PSO) algorithm is proposed to configure the orientations of IRS elements, while this research focuses on the single user and single transmitter scenarios.

This letter establishes an indoor multi-user IRS-aided VLC system, wherein one of the legitimate users is eavesdropped by an illegitimate user.
Based on the estimations of signal-to-interference-plus-noise ratio (SINR) level, a second-order polynomial-time algorithm is proposed to achieve the global optimization of the approximation problem, which finely approaches the original problem according to numerical results.
More insights are provided in theoretical analyses as well as the simulation part in the remainder of the letter.


\textit{Notation:}  
symbols $a$ ($A$), $\boldsymbol{a}$, and $\boldsymbol{A}$ represent the scalars, vectors, and matrices, respectively. 
Then, calligraphic letters $\mathcal{A}$ denote the defined index sets and $\mathbb{R}_+$ denote the real-valued and nonnegative number set.
Moreover, $\left(\cdot\right)^T$, $\mathbb{I}\left(\cdot\right)$, $\left\lfloor \cdot \right\rfloor$, and $\left\lceil \cdot \right\rceil$ denote the transpose operator, the indicator function, the floor function, and the ceil function, respectively.

\vspace{-0.4cm}
\section{System Model}
\vspace{-0.1cm}
\label{Sec:Model}
Consider a multi-user VLC system illustrated in Fig.~\ref{Fig:Diagram}, where Eve attempts to eavesdrop on a certain user and IRS is deployed to enhance system security. 
The light-of-sight (LoS) and Non-LoS (NLoS) channel models and received signals of both legitimate users and eavesdropper are discussed in the sequel of this section.

\vspace{-0.4cm}
\subsection{Channel Gain in Point Source Cases}
\vspace{-0.1cm}
\label{Subsec:Channel_gain}
\textit{1) 
LoS path:} 
In VLC, the LoS channel gain between the $k$-th user and the $l$-th light-emitting diode (LED) generally follows the Lambertian model as~\cite{pathak2015visible}
\begin{equation}
	\label{Eq:LoS_gain}
	\setlength\abovedisplayskip{3pt}
	h_{k,l}^{(1)} = \frac{A(m+1)}{2\pi d_{k,l}^2}\cos^m(\theta)g_{of}\cos(\phi)f(\phi),
	\setlength\belowdisplayskip{3pt}
\end{equation}
where $\theta$ is the angle of irradiance, $\phi$ is the angle of incidence, and $f(\phi)$ is formulated by refractive index $u$ and the semi-angle of the field-of-view (FoV) $\Phi$~\cite{pathak2015visible}.
Then, the parameters $A$, $m$, $g_{of}$, and $d_{k,l}$ represent the photodetector (PD) area, the Lambertian index, the optical filter gain, and the distance between the $l$-th transmitter and the $k$-th IRS unit, respectively.

\textit{2) 
NLoS path:} Given the negligible intensity level, the diffuse reflected light in IRS-aided VLC systems can generally be ignored~\cite{najafi2019intelligent,2021intelligent,abdelhady2020visible,wang2020performance}.
As for specularly reflected paths, some unique properties of IRS in the visible light range are listed as
\begin{itemize}
	\item The imaging method in geometric optics reveals that the reflected path can be regarded equivalently as emitted from the imaging transmitter~\cite{najafi2019intelligent,2021intelligent}, and therefore the direction of NLoS reflected path in IRS-aided VLC systems is easy to be controlled.
	
	\item Considering the nanoscale wavelength of the visible light, the near-field condition is guaranteed in the IRS-aided VLC systems according to~\cite{tang2020wireless}. 
	Consequently, the upper bound of irradiance level at PD follows an ``additive'' model under the point source assumption~\cite{abdelhady2020visible}, instead of the ``multiplicative'' one in the far-field case~\cite{tang2020wireless}. 
	
	\item In point source cases, one tiny IRS unit can only serve an individual user at a time since the propagation directions of the reflected paths rely on the specular reflection law strictly~\cite{sun_CL}, and the location difference between transmitters will lead to misalignment at the target PD.
\end{itemize}
 
Based on the aforementioned discussions, an upper bound of the irradiance performance is derived under the point source assumption as~\cite{abdelhady2020visible}
\begin{equation}
	\label{Eq:NLoS_gain}
	\setlength\abovedisplayskip{3pt}
	h_{k,n,l}^{(2)} = \delta\frac{A(m+1)}{2\pi \left(d_{n,l}+d_{k,n}\right)^2}\text{cos}^m(\theta)g_{of}\text{cos}(\phi)f(\phi),
	\setlength\belowdisplayskip{3pt}
\end{equation}
where $\delta$ is the reflectance of IRS unit, while $d_{n,l}$ and $d_{k,n}$ are the distances between the $l$-th LED and the $n$-th IRS unit and the $n$-th unit and the $k$-th user, respectively.

\begin{figure}[t]
	\centering
		\includegraphics[width=0.6\textwidth]{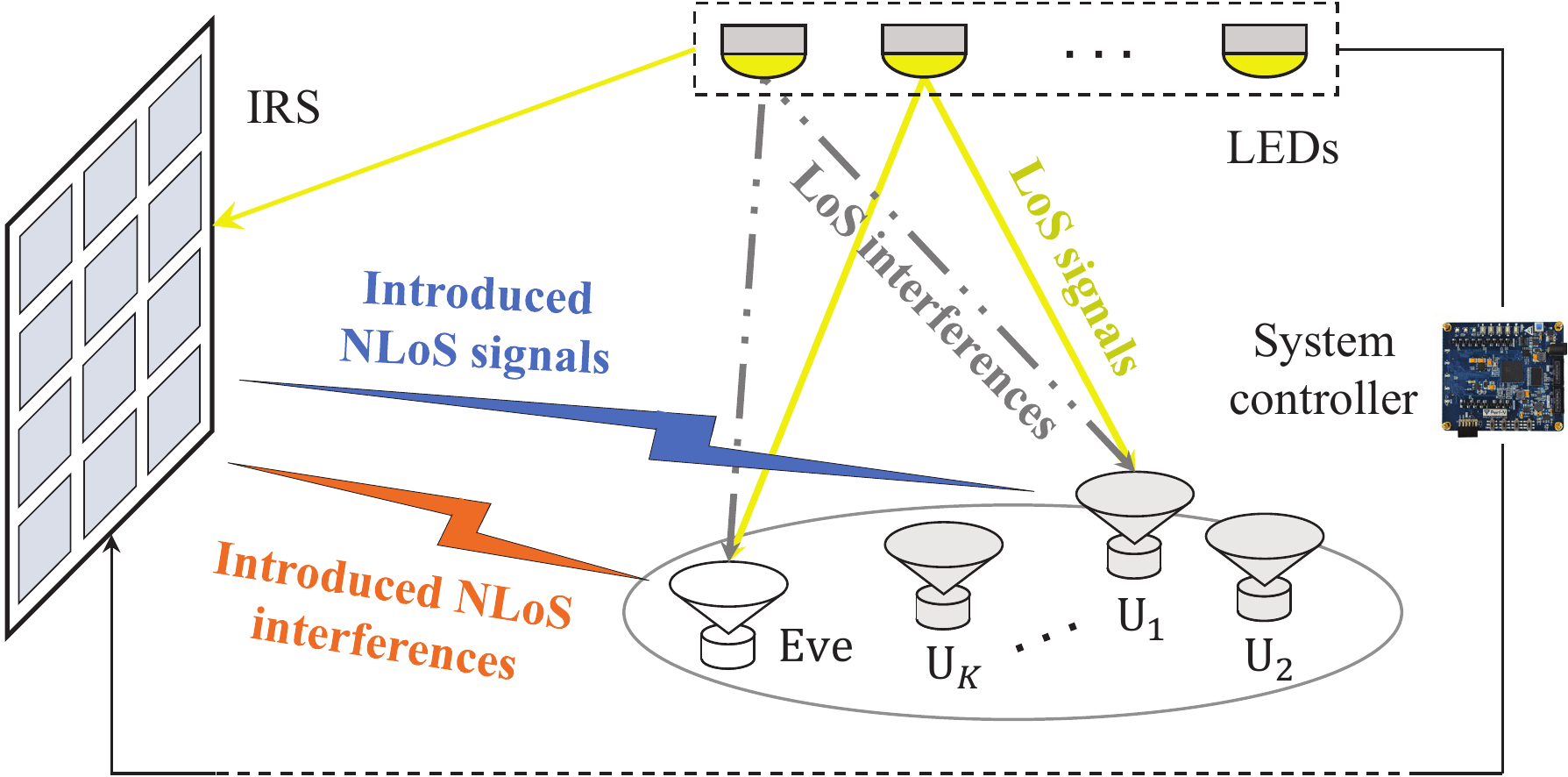}
	\caption{The system model of the secure IRS-aided VLC.}
	\label{Fig:Diagram}
	\vspace{-0.5cm}
\end{figure} 

\vspace{-0.35cm}
\subsection{Received Signals of Legitimate User and Eavesdropper}
\vspace{-0.1cm}
\label{Subsec:Received_signal}
When an individual user $k$ is served by the $l$-th transmitter, the received signal $\widehat{y}_{k,l}$ can be divided into three parts, namely the LoS component $\widehat{y}_{k,l}^{(1)}$, the NLoS component $\widehat{y}_{k,l}^{(2)}$, and the additive white Gaussian noise (AWGN) $z_k$.
Then, the total signal can be formulated as
\begin{equation}
	\label{Eq:Total_signal}
	\setlength\abovedisplayskip{3pt}
	\widehat{y}_{k,l} = \widehat{y}_{k,l}^{(1)} + \widehat{y}_{k,l}^{(2)} + z_k,
	\setlength\belowdisplayskip{3pt}
\end{equation}
where $z_k\sim\mathcal{N}(0,\sigma_k^2)$ with $\sigma_k^2$ the variance of the noise power.

\textit{1) Legitimate users:}
Suppose the transmission symbols on LEDs are represented by the vector $\boldsymbol{x}^T=\left[x_1,x_2,...,x_L\right]$, the LoS received signals of the $k$-th user are comprised of the intended information and the inter-user interferences as
\begin{equation}
	\setlength\abovedisplayskip{3pt}
	\label{Eq:LoS_intend}
	\widehat{y}_{k,l}^{(1)} = \rho_{k} h_{k,l}^{(1)}P_l x_{l} + \rho_{k} \sum_{i=1,i\neq l}^{L}h_{k,i}^{(1)}P_i x_{i},
	\setlength\belowdisplayskip{3pt}
\end{equation}
where $\rho_k$ and $P_l$ denote the PD responsivity and the emission power, respectively.
Without loss of generality, $x_l$ is independent with each other and with the expectation of 1.

Then, a binary matrix $\boldsymbol{G}=\left[\boldsymbol{g}_{1},\ \boldsymbol{g}_{2},\cdots,\ \boldsymbol{g}_{L} \right]_{N\times L}$ is defined to describe the association relationship between IRS units and LEDs, after which each unit can reconfigure itself based on a reverse lookup table~\cite{sun_CL}.
More specifically, the discrete element $g_{n,l}=1$ and $g_{n,l}=0$ indicate the cases that the $n$-th unit is and is not assigned to the $l$-th transmitter, respectively.
Therefore, the NLoS signal of the $k$-th legitimate user equals the aggregate gains of reflected paths that are related to IRS units of the $l$-th LED, which is given by
\begin{equation}
	\setlength\abovedisplayskip{3pt}
	\label{Eq:NLoS_intend}
	\widehat{y}_{k,l}^{(2)} = \rho_k \boldsymbol{h}_{k,l}^{(2)T} \boldsymbol{g}_l P_{l} x_{l},
	\setlength\belowdisplayskip{3pt}
\end{equation}
where a defined vector $\boldsymbol{h}_{k,l}^{(2)}=[h_{k,1,l},h_{k,2,l},...,h_{k,N,l}]^T \in \mathbb{R}_+^{N\times 1}$ is introduced to simplify the following discussions.

\textit{2) Eavesdropper:} 
Different from the traditional radio frequency (RF) communications, the intensity modulation and direct detection (IM/DD) scheme in VLC is based on the received light intensity, and the lack of phase information makes it impossible to eliminate the eavesdropper signal by passive beamforming techniques. 
Nevertheless, we attempt to degrade the SINR of the eavesdropper by purposely introducing an unintended interference, which can be implemented since the specularly reflected light is easy to be re-directed.
Compared to the legitimate users, the eavesdropper is considered as a special element with a vector $\boldsymbol{g}_0\in \left\{0,1\right\}^{N\times 1}$ recording its IRS assignment situation.
Once $g_{n,0}=1$ ensures, the $n$-th unit will adjust so that the unintended signal emitted from the LED closest to the eavesdropper can be reflected and then propagate to the eavesdropper.
As a consequence, the NLoS signal acts as interference and can be expressed as
\begin{equation}
	\setlength\abovedisplayskip{3pt}
	\label{Eq:NLoS_eve}
	\widehat{y}_{E,l}^{(2)} = \rho_E \boldsymbol{h}_{E,l_c}^{(2)T} \boldsymbol{g}_0 P_{l_c} x_{l_c},
	\setlength\belowdisplayskip{3pt}
\end{equation}
where $l_c$ is the index of the complementary transmitter to the $l$-th LED, i.e., the $l_c$-th LED carries the unintended signal and it has the shortest distance to the eavesdropper.

\vspace{-0.4cm}
\section{Optimization of Overall Secrecy Rate}
\vspace{-0.1cm}
\label{Sec:Proposed}
\subsection{Problem Formulation}
\vspace{-0.1cm}
\label{Subsec:Formulate}
Considering the constraints of the emission power limitation and real-value and nonnegative amplitude, the capacity-achieving input distribution in VLC is discrete instead of in a typical Shannon capacity form.
Nevertheless, a lower bound of the VLC dimmable channel capacity is proposed~\cite{wang2013tight}, through which the approximate capacity of the $k$-th legitimate user is given by
\begin{equation}
	\setlength\abovedisplayskip{3pt}
	\label{Eq:Rate_comp_intend}
	C_{k,l} = \frac{1}{2}W\log_2\left(1+\frac{e}{2\pi}\gamma_{k,l}\right),
	\setlength\belowdisplayskip{3pt}
\end{equation}
where $e$ and $W$ represent the value of the base of natural logarithms and the bandwidth, respectively.
The individual SINR of the $k$-th user is denoted as
\begin{equation}
	\setlength\abovedisplayskip{3pt}
	\label{Eq:SINR_intend}
	\gamma_{k,l} = \frac{\rho_k^2\left\{h_{k,l}^{(1)}+\boldsymbol{h}_{k,l}^{(2)T} \boldsymbol{g}_l\right\}^2P_l^2}{I_{k,l}},
	\setlength\belowdisplayskip{3pt}
\end{equation}
where the LoS interference plus noise can be expressed as
\begin{equation}
	\label{Eq:Noise_intend}
	\setlength\abovedisplayskip{3pt}
	I_{k,l} = \sigma_k^2+\rho_k^2\sum_{i=1,i\neq l}^{L}\left\{h_{k,i}^{(1)}P_{i}\right\}^2.
	\setlength\belowdisplayskip{3pt}
\end{equation}

This letter endeavors to maximize the expectation of multi-user secrecy rate for IRS-aided VLC systems.
To this end, the probability that the $k$-th user served by the $l$-th transmitter is denoted by a constant $f_{l,k}$ within the coherent time, which satisfies the equation $\sum_{k=1}^{K}f_{l,k} = 1$ according to the properties of the probability function.
Consequently, the upper bound of the eavesdropper capacity is given by
\vspace{-0.1cm}
\begin{align}
	\label{Eq:Rate_eve}
	C_E\! =\! \sum_{l=1}^L \frac{f_{l,k^*}W}{2}\log_2\left(\!1\!+\!\frac{e}{2\pi}\frac{\rho_{E}^2h_{{E},l}^{(1)2}P_l^2}{I_{{E},l}+\rho_{E}^2P_{l_c}^2\left\{\boldsymbol{h}_{{E},l_c}^{(2)T}\boldsymbol{g}_0\right\}^2}\right)\!,\notag\\
\end{align}
\vspace{-0.5cm}

\noindent
where $k^*$ is the user concerned by the eavesdropper.
To sum up, the overall secrecy rate is derived as~\cite{mostafa2015physical}
\begin{equation}
	\label{Eq:Secure_rate}
	\setlength\abovedisplayskip{3pt}
	C_S = \sum_{k=1}^K\sum_{l=1}^L f_{l,k}C_{k,l} - C_E,
	\setlength\belowdisplayskip{3pt}
\end{equation}
and the secrecy rate maximization problem is formulated as 
\vspace{-0.1cm}
\begin{align}
	\label{Pro:Ori_problem}
	\textbf{P}:\ \max\limits_{\widetilde{\boldsymbol{G}}}&\quad C_S\left(\widetilde{\boldsymbol{G}}\right) \\
	\label{Con:Discrete}
	\text{s.t.}\ 
	& \widetilde{g}_{n,l}\in\{0,1\}, \quad\  \forall n \in \mathcal{N}, l \in \mathcal{L}\cup\left\{0\right\},\\
	\label{Con:G_rowSum}
	& 
	\sum_{l=0}^{L}\widetilde{g}_{n,l}=1, \quad\  \forall n \in \mathcal{N},\\
	\label{Con:G_fair}
	&\sum_{n=1}^{N}\widetilde{g}_{n,l}\geq \left\lfloor \frac{N}{L+1} \right\rfloor,  \quad\  \forall l \in \mathcal{L}\cup\left\{0\right\},
\end{align}
\vspace{-0.2cm}

\noindent
where $\widetilde{\boldsymbol{G}}=\left[\boldsymbol{g}_0, \boldsymbol{G}\right]$ denotes an aggregate matrix, $\mathcal{L}$, $\mathcal{N}$, and $\mathcal{K}$ indicate the index sets of the transmitters, the IRS units, and legitimate users, respectively.
Then, the constraints in~(\ref{Con:Discrete}) and~(\ref{Con:G_rowSum}) come from the definitions of $\boldsymbol{G}$ and $\boldsymbol{g}_0$, and the constraint in~(\ref{Con:G_fair}) aims to guarantee the fairness of IRS configuration.

\vspace{-0.4cm}
\subsection{Proposed Algorithm to Maximize the Secrecy Rate}
\vspace{-0.1cm}
\label{Subsec:Proposed}
Mathematically, the problem \textbf{P} is an integer programming problem with the complexity of $\mathcal{O}(\left(L+1\right)^N)$, which is non-deterministic polynomial-time (NP) hard.
To reduce the complexity, a relaxing greedy algorithm is proposed to maximize the achievable sum rate of intended users~\cite{sun_CL}, wherein one of the subproblems is proved to be asymptotically convex.
However, such an optimization algorithm cannot be directly used in \textbf{P} since the term of eavesdropper rate destroys the structure of objective function.
Furthermore, considering the numbers of IRS units can be extremely large, an algorithm with controllable complexity is desirable in dealing with the time-variant environment.

\begin{figure}[t]
	\centering
	\subfigure[KM process]{
		\includegraphics[width=0.45\textwidth]{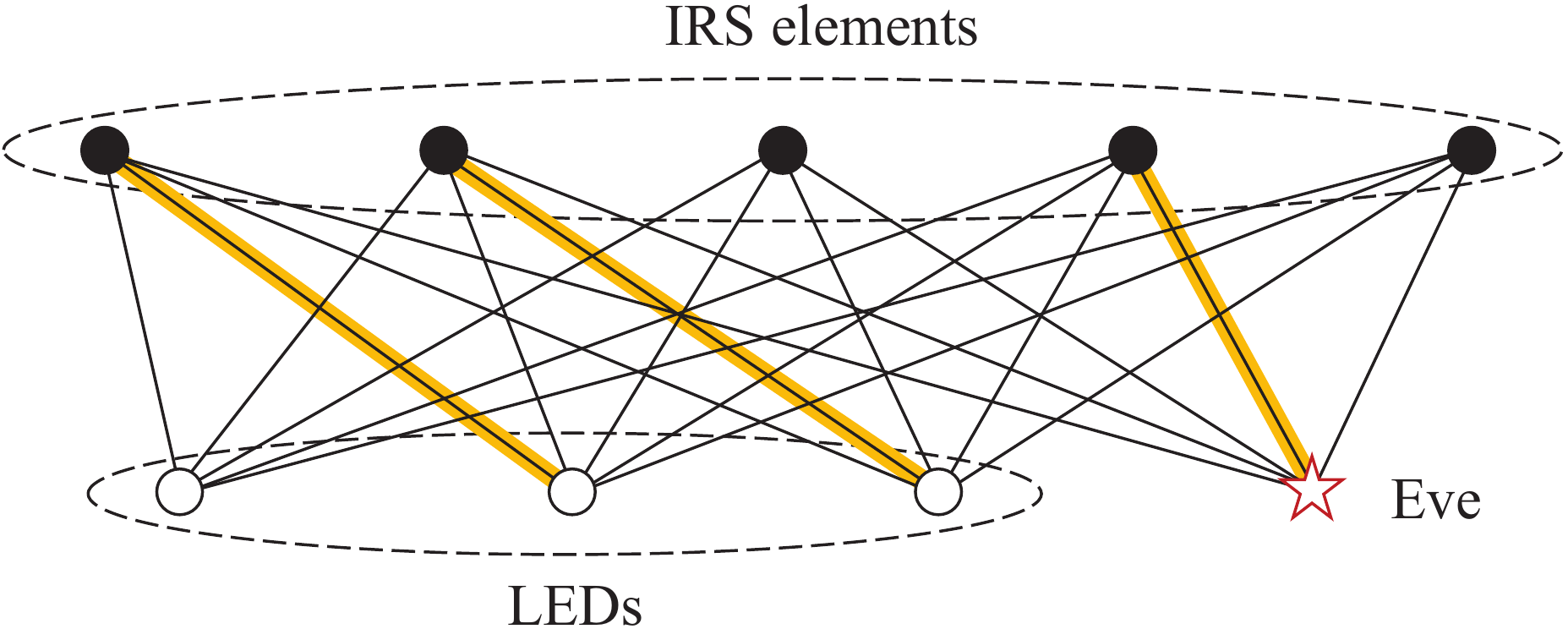}
	}
	\quad
	\subfigure[Delete the selected vertices]{
		\includegraphics[width=0.45\textwidth]{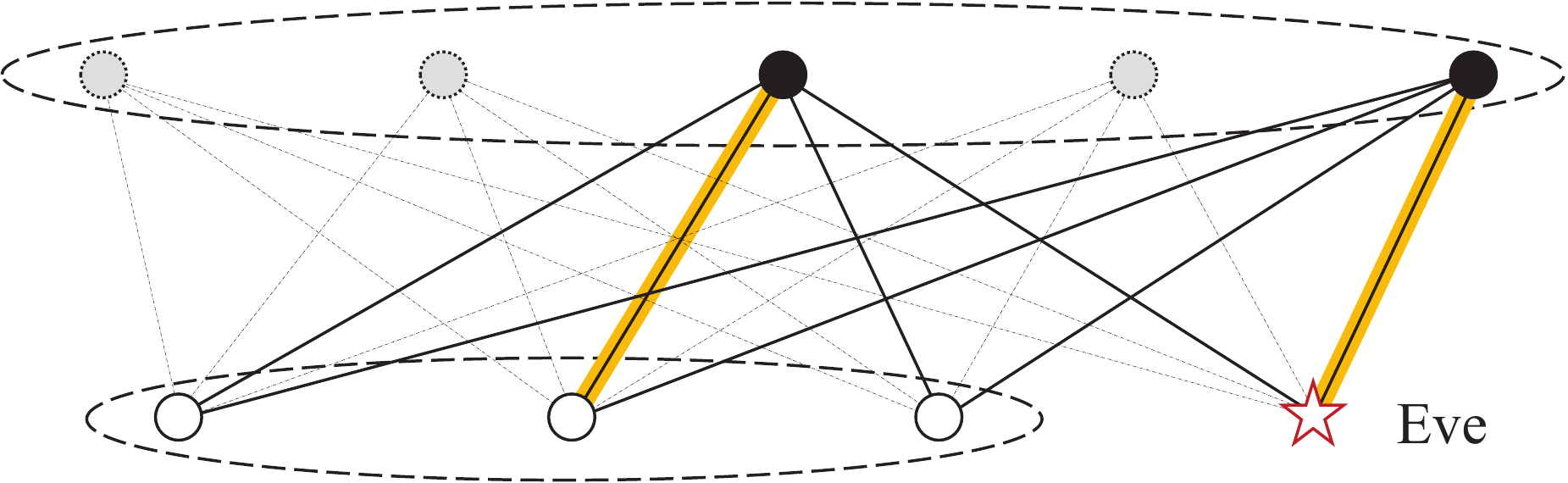}
	}
	\caption{The sketch map of the proposed iterative KM algorithm.} 
	\label{Fig:iterKM}
	\vspace{-0.5cm}
\end{figure} 

Based on the characteristics of the objective function and constraints, \textbf{P} can be reanalyzed from a discrete perspective, i.e., how to assign the $n$-th IRS unit to one of the $L+1$ targets (transmitters and eavesdropper) on the premise of fairness.
The limitation is that variables are coupled together in $C_S$, which leads to the difficulty to analyse the variation of secrecy rate.
Nevertheless, the capacity formulas in~(\ref{Eq:Rate_comp_intend}) and~(\ref{Eq:Rate_eve}) are exactly in the logarithmic form, which hints that they can be transformed into the form of linear combination by proper approximation.

\begin{lemma}
	\label{Le:Appro}
	The function $\eta\left(x^*\right)\log\left(x\right) + \xi\left(x^*\right)$ is a tight lower bound of $\log\left(1+x\right)$, where $\eta$ and $\xi$ are coefficients with respect to the tangent point $x^*$~\cite{papandriopoulos2008optimal}.
\end{lemma}

Given the functions in Lemma~\ref{Le:Appro}, the capacity of the legitimate user can be approximated at the given SINR $\gamma_{k,l}$ as
\vspace{-0.4cm}
\begin{align}
	\label{Eq:Rate_appro_intend}
	2C_{k,l}/W &\!\approx\ \! \eta_{k}\log_2\!\left(\frac{e}{2\pi}\frac{\rho_k^2\left\{h_{k,l}^{(1)}+\boldsymbol{h}_{k,l}^{(2)T} \boldsymbol{g}_l\right\}^2P_l^2}{I_{k,l}}\!\right)\! +\! \xi_{k} \notag\\
	&\!=\!  2\eta_{k}\log_2\left(1\!+\!\frac{\boldsymbol{h}_{k,l}^{(2)T} \boldsymbol{g}_l}{h_{k,l}^{(1)}}\right) \!+\! \eta_{k}\log_2\left(\Gamma_{k,l}\right) \!+\! \xi_{k},
\end{align}
\vspace{-0.3cm}

\noindent
where $\eta_k$ and $\xi_k$ are approximation functions of $\gamma_{k,l}$~\cite{papandriopoulos2008optimal} and $\Gamma_{k,l} = e\rho_{k}^2h_{{k},l}^{(1)2}P_l^2/\left(2\pi I_{{k},l}\right)$ denotes the LoS component of individual SINR.
Then, the first logarithmic term can be further approximated by Taylor expansion as
\vspace{-0.2cm}
\begin{align}
	\label{Eq:Rate_appro_intend_simplify}
	&C_{k,l} \approx \notag\\
	&\! W\!\eta_{k}\!\left(\!\frac{\boldsymbol{h}_{k,l}^{(2)T} \boldsymbol{g}_l-\lambda_kh_{k,l}^{(1)}}{h_{k,l}^{(1)}(1+\lambda_k)\ln2}\! +\! \log_2\!\left(1+\lambda_k\!\right)\!\right)\! +\! \frac{W}{2}\! \log_2\!\left(\!2^{\xi_{k}}\Gamma_{k,l}^{\eta_{k}}\!\right)\!,
\end{align}
\vspace{-0.3cm}

\noindent
where $\lambda_k = (\sum_{n=1}^N\sum_{l=1}^Lh_{k,n,l}^{(2)})/(L\sum_{l=1}^Lh_{k,l}^{(1)})$ is the tangent point, and the formula holds finely since the NLoS intensity level is weaker compared to the LoS one.
Similarly, the capacity of the eavesdropper can be rewritten as
\begin{align}
	\label{Eq:Rate_appro_eve}
	&-2C_E/W \notag\\
	&\approx\! -\!\sum_{l=1}^L \!f_{l,k^*}\!\left\{\eta_E\log_2\left(\frac{e\rho_{E}^2h_{{E},l}^{(1)2}P_l^2/\left(2\pi\right)}{I_{{E},l}+\rho_{E}^2P_{l_c}^2\left\{\boldsymbol{h}_{{E},{l_c}}^{(2)T}\boldsymbol{g}_0\right\}^2}\right)\! +\! \xi_E\right\}\notag\\
	&=\! \sum_{l=1}^L\!f_{l,k^*}\!\left\{\!\eta_E\log_2\left(1\!+\!\frac{\left\{\boldsymbol{h}_{{E},{l_c}}^{(2)T}\boldsymbol{g}_0\right\}^2}{I_{{E},l}/\rho_{E}^2/P_{l_c}^2}\right)\!-\!\log_2\left(2^{\xi_E}\Gamma_{E,l}^{\eta_E}\right)\!\right\} \notag\\
	&\overset{(a)}{\approx} \sum_{l=1}^Lf_{l,k^*}\left\{\frac{\eta_E\rho_{E}^2P_{l_c}^2\Delta}{I_{{E},l}\ln2}\boldsymbol{h}_{{E},{l_c}}^{(2)T}\boldsymbol{g}_0 -\log_2\left(2^{\xi_E}\Gamma_{E,l}^{\eta_E}\right)\right\},
\end{align}
\vspace{-0.3cm}

\noindent
where $\left(a\right)$ satisfies due to the first-order Taylor expansion at the point 0 and $\Delta$ is a constant quantity of the same order of magnitude as $\boldsymbol{h}_{{E},{l_c}}^{(2)T}\boldsymbol{g}_0$.
With proper approximations, the overall secrecy rate can be linearly divided into the NLoS components and LoS direct bias as
\begin{equation}
	\label{Eq:Re_construct}
	\setlength\abovedisplayskip{3pt}
	\widehat{C}_S\left(\widetilde{\boldsymbol{G}}\right) = \sum_{n=1}^N\sum_{l=0}^L w_{n,l}\cdot \widetilde{g}_{n,l} + Q,
	\setlength\belowdisplayskip{3pt}
\end{equation}
where the bias term $Q$ results from the constant parts in~(\ref{Eq:Rate_appro_intend}) and~(\ref{Eq:Rate_appro_eve}), and the coefficients $w_{n,l}$ can be formulated as
\vspace{-0.2cm}
\begin{align}
	\label{Eq:Weight}
	w_{n,l} =\ & \mathbb{I}\left(l>0\right)\sum_{k=1}^K\frac{\eta_kWf_{l,k}h_{k,n,l}^{(2)}}{h_{k,l}^{(1)}\ln2} \notag\\
	&+ \mathbb{I}\left(l=0\right)\sum_{i=1}^L\frac{\eta_Ef_{i,k^*}W\rho_{E}^2P_{i_{c}}^2h_{E,n,i_{c}}^{(2)}}{2I_{{E},i}\ln2}\Delta.
\end{align}
\vspace{-0.4cm}

\noindent
where the index of the complementary transmitter to the $i$-th LED is denoted by $i_c$, and $\Delta=\sum_{n=1}^N\sum_{l=1}^Lh_{E,n,l}^{(2)}/L^2$.

\vspace{-0.05cm}
\begin{algorithm}[t]
	\caption{Proposed Iterative KM Algorithm}
	\label{Alg:Solve_P1}
	\hspace*{0.02in} {\bf Input:} $h_{k,l}^{(1)}$, $h_{k,n,l}^{(2)}$, $t \gets 0$.\\
	\hspace*{0.02in} {\bf Output:} $\widetilde{\boldsymbol{G}}$.
	\begin{algorithmic}[1]
		\REPEAT
		\STATE Calculate ($\eta_{k}$, $\xi_{k}$, $\eta_{E}$, $\xi_{E}$) as Lemma~\ref{Le:Appro} and $t \gets 0$;
		\STATE Calculate the rate bias $Q$ according to~(\ref{Eq:Rate_appro_intend}) and~(\ref{Eq:Rate_appro_eve});
		\STATE Generate the weight matrix $W$ according to~(\ref{Eq:Weight});
		\REPEAT
		\STATE Run the KM algorithm;
		\STATE $g_{n,l} \gets 1$ for the selected edges;
		\STATE Delete these vertices and their related edges;
		\STATE $t\gets t+1$
		\UNTIL{$N \leq t\left(L+1\right)$}
		\UNTIL{\textit{Convergence}}
	\end{algorithmic}
\end{algorithm}
\setlength{\textfloatsep}{0.1cm}

Consequently, the secrecy rate maximization process is transformed into an optimal matching search in a bipartite graph, wherein two index sets are $\mathcal{N}$ and $\mathcal{L}\cup\left\{0\right\}$.
The Kuhn-Munkres (KM) algorithm can generally achieve the global optimal solution of such bipartite problems, but the maximum number of matching edges in the KM algorithm is $\min\left(N, L+1\right)$.
To ensure the fairness constraint in~(\ref{Con:G_fair}), we propose an iterative KM algorithm to solve the modified \textbf{P}.
As shown in \textbf{Algorithm}~\ref{Alg:Solve_P1}, the weight matrix $W$ and rate bias term $Q$ are calculated before the assignment.
Then, the KM algorithm is carried out in each loop $t$, after which the indices of selected IRS units and their corresponding edges are deleted as illustrated in Fig.~\ref{Fig:iterKM}.
This process will go on till the number of remaining IRS units is a negative number, and the above steps correspond to one time of assignment.
Considering the IRS configuration result affects the individual SINR conversely, the unit assignment and approximate parameters need to be conducted alternatively until convergence, i.e., the matrix $\widetilde{\boldsymbol{G}}$ remains unchange.
Moreover, the proposed algorithm will inevitably end since there are at most $\left\lceil N/\left(L+1\right) \right\rceil$ KM algorithm calls in each assignment.

\vspace{-0.7cm}
\subsection{Analyses on Optimality and Complexity}
\vspace{-0.2cm}
\label{Subsec:Analyses}
\textit{1) Global optimality analysis:} 
The result of \textbf{Algorithm}~\ref{Alg:Solve_P1} will naturally satisfy the fairness constraint~(\ref{Con:G_fair}) due to the iterative process.
Then, the analysis on the optimality of the proposed algorithm is given as follows.
\begin{lemma}
	\label{Le:KMbest}
	The KM algorithm obtains the optimal matching result of the weighted bipartite graph~\cite{6882839}.
\end{lemma}
\begin{proposition}
	The proposed iterative KM algorithm will achieve the optimal result of the approximate rate function.
\end{proposition}
\begin{proof}
If each element in $\mathcal{L}\cup\left\{0\right\}$ has up to one matching edge, the proposed algorithm will achieve the optimal result according to Lemma~\ref{Le:KMbest}.
Suppose the optimality ensures when the number of matching edges for an element is $J$, which means that the total matched edges is $J\left(L+1\right)$.
Then, when the number goes to $J+1$, the optimality of the proposed algorithm still holds because the KM algorithm can obtain the optimal matching result in each assignment.
\end{proof}

\textit{2) Complexity analysis:}
According to~(\ref{Eq:Weight}), the calculations of the weighted matrix and direct bias leads to a complexity about $\mathcal{O}\left(NLK\right)$.
In the $t$-th loop, the KM algorithm is processed on a bipartite graph composed of two sets of $L+1$ points and $N-t\left(L+1\right)$ points.
Each loop has a complexity about $\mathcal{O}(N\left(L+1\right)^2-t\left(L+1\right)^3)$~\cite{6882839} and the proposed algorithm conducts till $t=\left\lceil N/\left(L+1\right) \right\rceil$.
Therefore, the computational complexity of \textbf{Algorithm}~\ref{Alg:Solve_P1} is given by
\vspace{-0.2cm}
\begin{align}
	\label{Eq:Complexity}
	&\sum_{t=0}^{\left\lceil N/\left(L+1\right) \right\rceil-1} \left\{N\left(L+1\right)^2-t\left(L+1\right)^3\right\} T \notag\\
	&\qquad = T\left(L+1\right)^2 \left\lceil\frac{N}{L+1}\right\rceil \left\{N-\frac{L+1}{2}\left(\left\lceil\frac{N}{L+1}\right\rceil-1\right)\right\} \notag\\
	&\qquad \approx \frac{TNL\left(N+L\right)}{2},
\end{align}
\vspace{-0.5cm}

\noindent
where $T$ denotes the time of assignments.
Notably, the time consumption is proportional to the quadratic power of $N$ and $L$, which is far lower than the exhaustive search method with the complexity of $\mathcal{O}(\left(L+1\right)^N)$.

\begin{table}[t]
	\centering
	\caption{Simulation Parameters}
	\vspace{-0.2cm}
	\label{parameters}
	\begin{tabular}{| c | c | c | c |}
		\hline
		$K = 4$ & $L = 4$ & $W = 20\ \text{MHz}$ & $\delta = 0.5$\\
		\hline
		$g_{of} = 1$ & $m = 1$ & $A = 4\ \text{cm}^2$ & $\rho_k = 0.5\ A/W$\\
		\hline
		$\Phi = 80^{\circ}$ & $u = 1.5$ & $D = 100\ \text{cm}^2$ & $\sigma^2 = 10^{-10}\ W$\\
		\hline
	\end{tabular}
	\label{Tab:SimuPara}
\end{table}

\vspace{-0.5cm}
\section{Numerical Results}
\vspace{-0.2cm}
\label{Sec:Numerical}
In this section, we provide simulation results to testify the previous theoretical analyses.
Specifically, four LEDs located at (1m, 1m, 3m), (1m, 7m, 3m), (7m, 1m, 3m), and (7m, 7m, 3m) are transmitters in an 8m $\times$ 8m $\times$ 3m room, where users walk randomly in the plane 0.5m above the ground.
For each transmitter, its service probability for a certain user is inversely proportional to the square of LoS propagation distances.
Then, a planar IRS with unit area $D$ is deployed on the wall, and its horizontal and vertical margins are 1m and 0.3m, respectively. 
Moreover, the spacing between two IRS units is set as 20cm, and more detailed parameters are given in Table~\ref{Tab:SimuPara}.

To start with, the numerical simulation is executed to evaluate the performance of the proposed algorithm as well as the correctness of functions approximation in~(\ref{Eq:Rate_appro_intend}) and~(\ref{Eq:Rate_appro_eve}).
Without loss of generality, four users are located at (3.6m, 2.7m, 0m), (1.0m, 3.3m, 0m), (3.0m, 4.5m, 0m), and (6.4m, 2.2m, 0m), where the first one is eavesdropped by Eve located at (2.1m, 1.5m, 0m).
The emission power on four LEDs varies from 0 dBW to 10 dBW, and different benchmark schemes are adopted for comparison:

\textit{1) \textbf{Approximation secrecy rate:}} 
Assign IRS units according to \textbf{Algorithm}~\ref{Alg:Solve_P1} and obtain the approximate result $\widehat{C}_S(\widetilde{\boldsymbol{G}})$.

\textit{2) \textbf{Proposed algorithm with \& without Eve SINR:}} 
Assign IRS units according to \textbf{Algorithm}~\ref{Alg:Solve_P1}, and the secrecy rate functions in~(\ref{Eq:Rate_appro_eve}) are approximated at real SINR and a random point, respectively.
The secrecy rate is calculated by~(\ref{Eq:Secure_rate}).

\textit{3) \textbf{Random assignment:}} 
Assign all IRS units equally among different columns so that each of them has nearly $N/\left(L+1\right)$ ones.
Then, the vaiable $\widetilde{\boldsymbol{G}}$ is scrambled and randomly rearranged according to rows.
For accuracy, the secrecy rate is calculated by averaging 200 independent trials.

\textit{4) \textbf{No IRS:}} Obtain the secrecy rate with $\widetilde{\boldsymbol{G}}=\boldsymbol{0}$.


As shown in Fig.~\ref{Fig:Power}, the proposed approximation method is considerably tight to the capacity formula in~(\ref{Eq:Secure_rate}), and the rate gap becomes large when $N$ increases from 8 to 64.
This is because the step in~(\ref{Eq:Rate_appro_intend_simplify}) and~(\ref{Eq:Rate_appro_eve}) generate non-negligible errors with the increase of $N$.
Notably, the VLC system security benefits a lot from the deployed IRS. 
When the emission power is 10 dBW, the surface with $N=64$ units can achieve nearly 30 Mbps gain compared to the case without IRS, and even the random assignment scheme improves the secrecy rate by 15 Mbps.
Nevertheless, the rate gain obtained by IRS is much smaller when $N=8$, i.e., less than 8 Mbps at 10 dBW power, which reveals that the secrecy improvement performance is sensitive to the number of IRS units.
Moreover, the results also show that the proposed algorithm has less SINR requirement for the eavesdropper, namely the tangent point of $\eta_{E}$ and $\xi_{E}$ will not significantly affect the secrecy rate.
This can be explained by the closeness between the approximate function and the original rate function in Lemma~\ref{Le:Appro}.

\begin{figure}[t]
	\centering
	\includegraphics[width=0.55\textwidth]{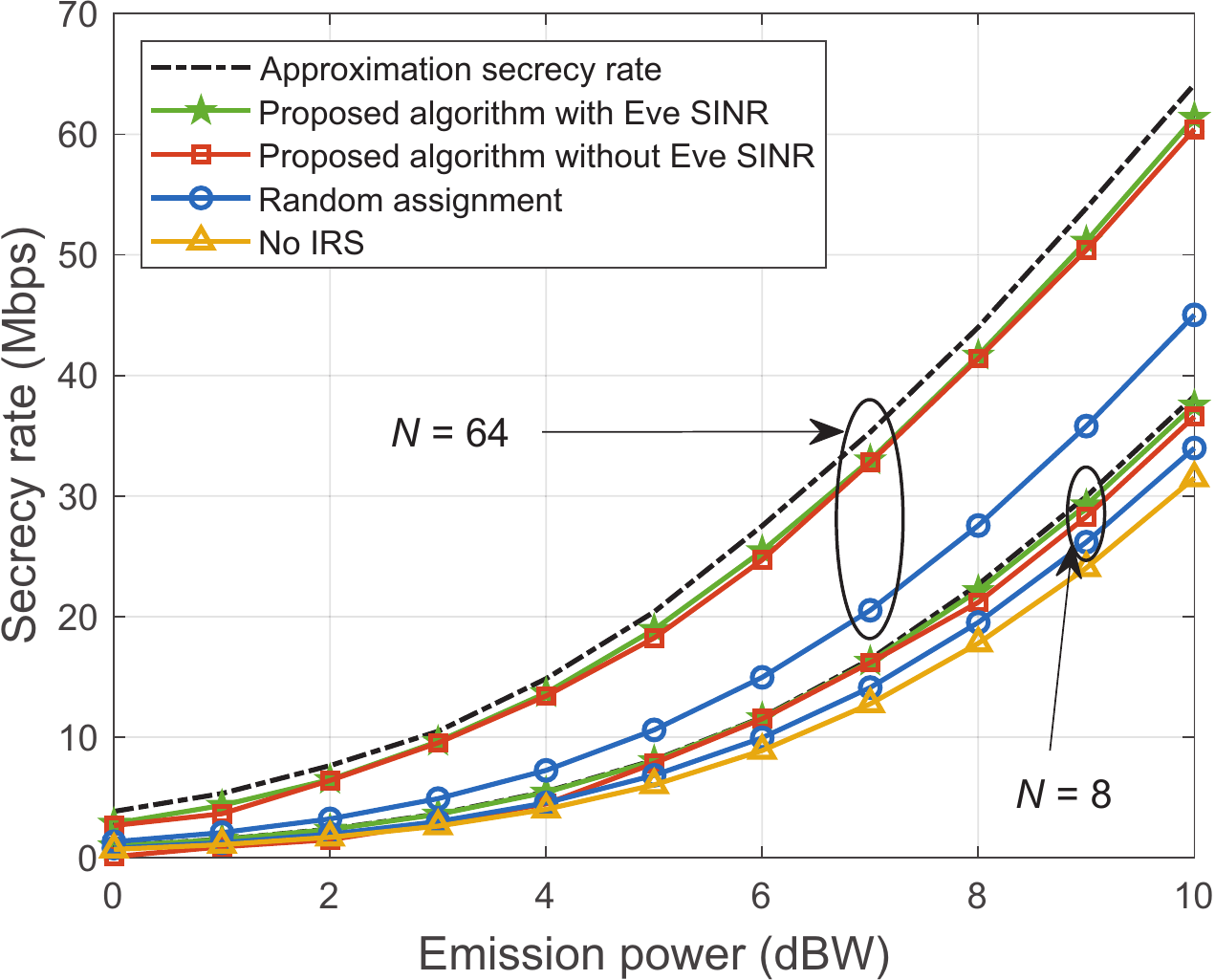}
	\caption{The performance of proposed algorithm compared with other baselines.} 
	\label{Fig:Power}
\end{figure} 
\begin{figure}[t]
	\centering
	\includegraphics[width=0.55\textwidth]{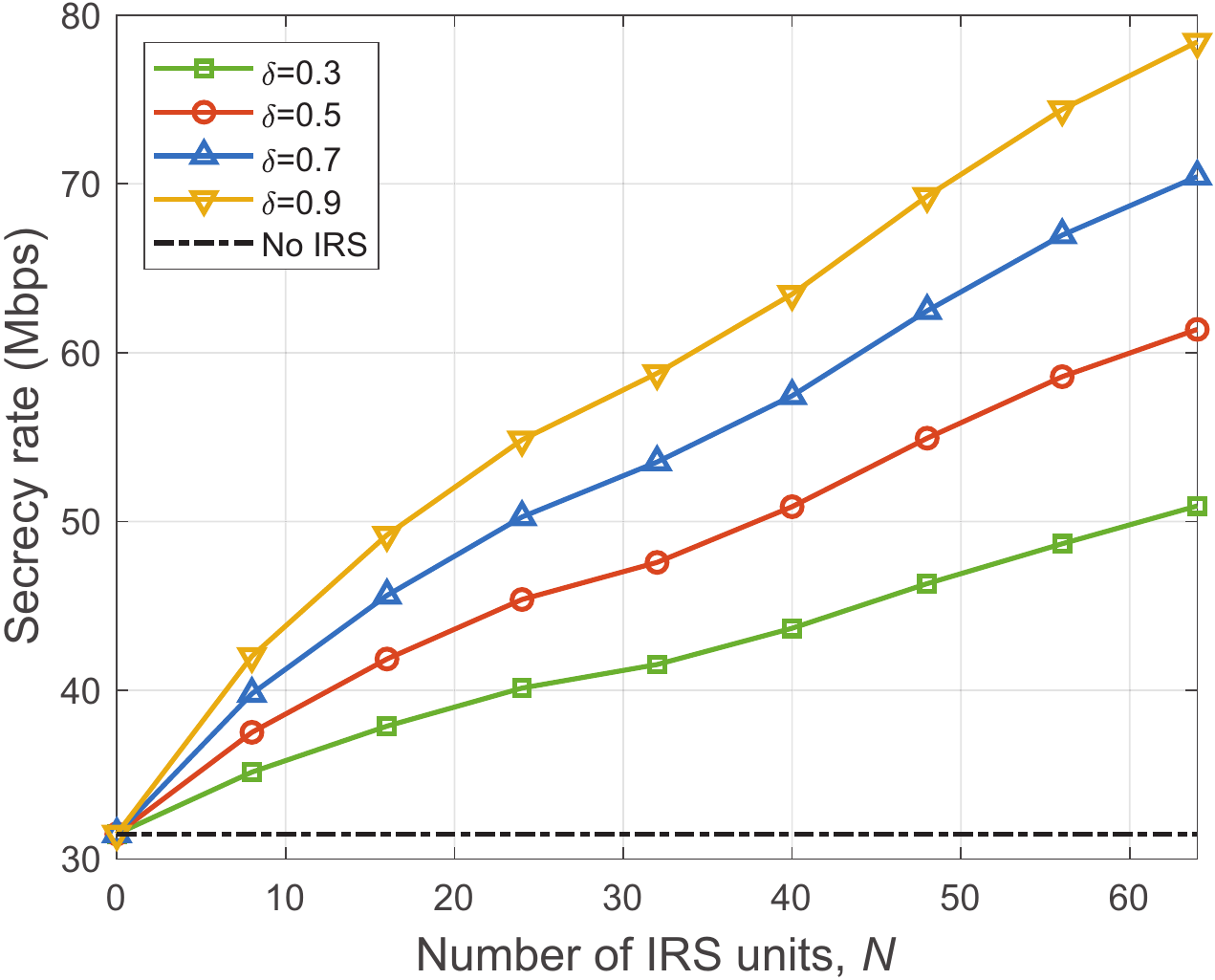}
	\caption{The secrecy rate versus the number of IRS units under the different value of reflectance.} 
	\label{Fig:App}
\end{figure}

Numerical simulations are carried out in Fig.~\ref{Fig:App} to investigate the influence of reflectance as well as the number of IRS units.
The emission power herein is 10 dBW, and $N$ increases from 0 to 64 under certain reflectance values.
It can be seen from the result that there is a positive correlation between the overall secrecy rate and $N$, e.g., the rate when $\delta=0.5$ and $N=64$ is doubled compared to no IRS case.
Besides, the numerical results show that a higher secrecy rate can also be achieved with large $\delta$, i.e., the case with $\delta=0.7$ and $N=8$ and the one with $\delta=0.3$ and $N=24$ both correspond to the overall secrecy rate of 40 Mbps.
This phenomenon suggests that the secrecy rate can be increased by cooperatively determining the values of reflectance and the number of IRS units.

\vspace{-0.3cm}
\section{Conclusions}
\vspace{-0.2cm}
\label{Sec:Conclu}
An indoor IRS-aided secure VLC system is modeled in this letter, wherein one of the legitimate users is eavesdropped.
Through the approximation of the objective function, the secrecy rate maximization problem is transformed into an assignment problem, and an iterative KM algorithm with second-order polynomial complexity is proposed to search for the optimal result.
Numerical results show that the secrecy rate has been prominently improved by IRS, and the rate gain is nearly proportional to the number of units and the reflectance value.
Moreover, the usage of IRS offers a promising research direction for VLC physical layer security.

\vspace{-0.3cm}
\def\bibfont{\fontsize{7.8}{9.3}\selectfont}
\bibliographystyle{abbrv}
\bibliography{IEEEabrv,reference}
\vspace{-0.4cm}
\end{document}